\documentclass{amsart}
\usepackage{mathrsfs,amssymb,amsmath,amsfonts,amsxtra}
\usepackage{verbatim}
\usepackage{graphicx,color}

\begin{document}
\sloppypar \sloppy
\title{On the Simulation of Adaptive Measurements via Postselection}
\author{Vikram Dhillon}
\email{dhillonv10@gmail.com}%

\maketitle

\begin{abstract}

In this note we address the question of whether any any quantum computational model that allows 
adaptive measurements can be simulated by a model that allows postselected measurements. We argue in the favor of this question and prove that adaptive measurements 
can be simulated by postselection. We also discuss some potentially stunning consequences of this result such as the ability to solve \#P problems.  

\end{abstract}

\newcommand{\pd}{\partial} 
\newcommand{\ud}{\mathrm{d}} 

\newcommand{\B}{\mbox{$\mathbb{Z}_2$}} 
\newcommand{\Z}{\mbox{$\mathbb{Z}$}}
\newcommand{\R}{\mbox{$\mathbb{R}$}}
\newcommand{\C}{\mbox{$\mathbb{C}$}}

\newtheorem{theorem}{Theorem}
\newtheorem{axiom}{Axiom}
\newtheorem{definition}[axiom]{Definition}
\newtheorem{lemma}[axiom]{Lemma}

\def\boldclass{\bf\sf}
\def\P{{\boldclass P}}
\def\NP{{\boldclass NP}}
\def\PH{{\boldclass PH}}
\def\PP{{\boldclass PP}}
\def\BPP{{\boldclass BPP}}
\def\BQP{{\boldclass BQP}}
\def\PostBQP{{\boldclass PostBQP}}
\def\PostBPP{{\boldclass PostBPP}}
\def\PSPACE{{\boldclass PSPACE}}
\def\EXP{{\boldclass EXP}}
\def\FPTAS{{\boldclass FPTAS}}
\def\BPPpath{{\boldclass BPP_{path}}}
\def\PPP{{\boldclass P^{PP}}}
\def\PSHARP{{\boldclass P^{\#P}}}

\section{Introduction}

In \cite{Aaronson} Aaronson introduced a complexity class $\PostBQP$, which is is a complexity class consisting of all of the computational
problems solvable in polynomial time on a quantum Turing machine with postselection and bounded error. It was also shown equivalent to 
$\PP$ which is the class of decision problems solvable by a probabilistic Turing machine in polynomial time, with an error probability of
less than 1/2 for all instances. Aaronson then raised an interesting question which asks whether adaptive measurements made by a quantum
computational model be simulated with postselected measurements. In this note we address this question by asserting that it is possible 
to simulate adaptive measurements by postselection on the quantum circuit model of computation. We also explore the consequences of being 
able to atleast theoritically perform this simulation, it is known that $ \PPP = \PSHARP $ which implies that the complexity of $\PP$ is 
equivalent to that of $\PSHARP$ which is an $\NP$.Sso if an adaptive (non-projective) measurement such as a weak measurement can be simulated, 
following the work of Lloyd $\mathit{et. al}$ logical gates can be constructed that allow us to solve $\PSHARP$ problems.

\section{Proof} 

Before we show how the simulation would work, we want to establish some definitions to make an easier transition to the proof itself. 

\begin{definition}
An adaptive measurement is an incomplete measurement is made on the system, and its result used
to choose the nature of the second measurement made on the system,
and so on (until the measurement is complete). A complete measurement
is one which leaves the system in a state independent of its initial
state, and hence containing no further information of use. \cite{One}
\end{definition}

\begin{definition}
Postselection is the power of discarding all runs of a computation in which a given event does not occur. \cite{Aaronson}
\end{definition}

We will be using the quantum circuit model which is the standard model in quantum computation theory and most other computational models have been
shown to be equivalent to it. The equivalence also allows us to simulate those models on the circuit model. 

\begin{lemma}
Measurement based quantum computation (MBQC) employs adaptive local measurements on a resource state. 
\end{lemma}

\begin{proof}
See \cite{One} for this. 
\end{proof}

\begin{lemma}
Measurement based models can be simulated on the quantum circuit model.
\end{lemma}

\begin{proof}
Any one-way computation can be made into a quantum circuit by using quantum gates to prepare the resource state \cite{MSBC}.
\end{proof}

\begin{axiom}
From Lemma 3 and Lemma 4 we can deduce that the qunatum circuit model can simulate measurement based computation which is a model that allows for adaptive measurements.
\end{axiom}

The above mentioned axiom completes the first part of the corrospondence, we now have to show that the same model that can simulate postselected measurements to complete the 
corrospondance. Postselected measurements fall under the complexity class $\PostBQP$ and we will also use the equivalence of $\PostBQP$ and $\PP$ shown by Aaronson in 
\cite{Aaronson}. This switch between complexity classes makes this proof simplistic.

\begin {axiom}
$\BQP \subset \PP$
\end{axiom}

\begin{lemma}
$\PP \cap \BQP \notin \{\phi\} $
\end{lemma}

\begin{proof}
Let us assume that no problem exist at the intersection of $\PP$ and $\BQP$ However, we know the aforementioned axiom to be true so there must atleast be one problem that exist at the 
intersection of the complexity classes. That particular problem, by the virtue of being at the intersection will be both $\PP$ and $\BQP$ which is self-evident. We will 
represent the problems present at the intersection of the two complexity classes by the set $\tau$.
\end{proof}

\begin{lemma}
From Lemma 7 we can deduce that elements of $\tau$ can be simulated on a quantum computer
\end{lemma}

\begin{proof}
The elements of $\tau$ fall in the clas $\BQP$ which can be simulated on a quantum computer therefore the elements of that set can also be simulated by a quantum computation
model. 
\end{proof}

\begin{lemma} 
Elements of $\tau$ can exibit postselected measurements
\end{lemma}

\begin{proof} 
The members of $\tau$ are both $\PP$ and $\BQP$ where $\BQP$ can be simulated on a quantum computer and since $\PP = \PostBQP$, elements of $\tau$ can be simulated through
the use of postselected measurements.
\end{proof}

\begin{axiom} 
From the preceding proof and Axiom 5, we see the corrospondence is complete. We can indeed take a quantum computational model (in our case, it is the standard quantum circuit model)
that allows for adaptive measurements (Axiom 5) and simulate it with a model that allows for postselected measurements (again the quantum circuit model)
\end{axiom} 

The preceeding axiom presents the completed proof, in the following section we will discuss some speculative consequences of this result. 

\section{Consequences} 

Although machines capable of postselectd measurements are implausable, the ability to simulate a particular type of adaptive measurement called weak measurement has some 
very interesting consequences. Lloyd $\mathit{et. al}$ \cite{Lloydweak} showed that when weak measurements are made on a set of identical quantum systems, the single-system 
density matrix can be determined to a high degree of accuracy while affecting each system only slightly. If this information can then be fed back into the system by 
coherent operations, the single-sytem density matrix can be made to undergo arbitrary nonlinear dynamics such as dynamics governed by a nonlinear Schrödinger equation.
Nonlinear corrections to quantum mechanical evolution can then be used to construct nonlinear quantum gates which can solve $\#\P$ and $\NP$-Complete problems as shown by 
Lloyd and Abrams \cite{LloydAbrams}.


\begin{thebibliography}{99}

\bibitem{Lloydweak} Lloyd, Seth and Slotine, Jean-Jacques E. Quantum feedback with weak measurements. Phys. Rev. A (2000).

\bibitem{LloydAbrams} Abrams, Daniel S. and Lloyd, Seth. Nonlinear Quantum Mechanics Implies Polynomial-Time Solution for $\NP$-Complete and $\#\P$ Problems. Phys. Rev. Lett. (1998).

\bibitem{Aaronson} Taylor, P.L., Heinonen, O.: A Quantum Approach to Condensed 

\bibitem{One} R. Raussendorf, D. E. Browne, and H. J. Briegel. Measurement based Quantum Computation on Cluster States. Phys. Rev. A (2003). 

\bibitem{MSBC} H. M. Wiseman, D. W. Berry, S. D. Bartlett, B. L. Higgins, and G. J. Pryde. Adaptive Measurements in the Optical Quantum
Information Laboratory. IEEE Journal of Selected Topics in Quantum Electronics (2009).

\end{thebibliography}
\end{document}